\documentclass[journal,letterpaper,twocolumn,twoside,nofonttune]{IEEEtran}

\usepackage{multirow}

\usepackage{algorithm, algorithmic}

\usepackage[utf8]{inputenc} 
\usepackage[T1]{fontenc}
\usepackage{url}
\usepackage{ifthen}
\usepackage{cite}
\usepackage{amsmath}
\usepackage{wasysym}%
\usepackage{amssymb}
\let\labelindent\relax
\usepackage{enumitem}
\usepackage{mdwmath}
\usepackage{blindtext}
\usepackage{eqparbox}
\usepackage{stfloats}

\newcommand\mydots{\hbox to 0.7em{.\hss.\hss.}}

\usepackage[english]{babel}
\usepackage{lipsum}
\usepackage{xcolor}
\usepackage{tikz}

\usepackage{subfigure}
\usetikzlibrary{matrix,shapes,arrows,positioning,chains, calc}
\usepackage{graphicx}
\usepackage{geometry}
\newtheorem{theorem}{{Theorem}}
\newtheorem{lemma}[theorem]{{Lemma}}
\newtheorem{proposition}[theorem]{{Proposition}}
\newtheorem{corollary}[theorem]{{Corollary}}
\newtheorem{claim}[theorem]{{Claim}}
\newtheorem{definition}{{Definition}}

\DeclareMathAlphabet{\mathbfsl}{OT1}{ppl}{b}{it} 

\newcommand{\bc}{\mathbfsl{c}} 
 
\newcommand{\be}{\mathbfsl{e}} 

\newcommand{\bh}{\mathbfsl{h}}

\newcommand{\bu}{\mathbfsl{u}} 
\newcommand{\bv}{\mathbfsl{v}}
 
\newcommand{\bx}{\mathbfsl{x}}
\newcommand{\by}{\mathbfsl{y}} 
\newcommand{\bz}{\mathbfsl{z}}
\newcommand{\bk}{\mathbfsl{k}} 
\newcommand{\bm}{\mathbfsl{m}} 
 
\DeclareMathOperator*{\argmin}{argmin}

\newcommand*{\argminl}{\argmin\limits}

\newcommand{\floor}[1]{\left\lfloor #1 \right\rfloor}

\newcommand{\beq}[1]{\begin{equation}\label{#1}}
	\newcommand{\ee}{\end{equation}} 
\newcommand{\eq}[1]{(\ref{#1})}

\renewcommand{\leq}{\leqslant}
 
\renewcommand{\geq}{\geqslant}

\renewcommand{\Bbb}{\mathbb}

\newcommand{\Tref}[1]{Theo\-rem\,\ref{#1}}

\newcommand{\Cref}[1]{Co\-ro\-lla\-ry\,\ref{#1}}

\newcommand{\Fq}{{{\Bbb F}}_{\!q}}

\newcommand{\deff}{\mbox{$\stackrel{\rm def}{=}$}}

\DeclareMathOperator{\rank}{rank}

\begin{document}
\title{Threshold-Secure Coding with Shared Key} 

\author{\large%
Nasser Aldaghri,\,\,\IEEEmembership{Student Member,~IEEE}, and
Hessam Mahdavifar,\,\,\IEEEmembership{Member,~IEEE}\\

\thanks{%
The material in this paper was presented in part at the 57th Annual Allerton Conference on Communication, Control, and Computing (Allerton) in September 2019 \cite{aldaghri2019threshold}. This work was supported in part by the National Science Foundation under Grant CCF–1763348 and Grant CCF–1909771.}
\thanks{N.\, Aldaghri and H.\ Mahdavifar are with the Department of Electrical Engineering and Computer Science, University of Michigan, Ann Arbor, MI 48109 (email: aldaghri@umich.edu and hessam@umich.edu).} 
}

\maketitle
\begin{abstract}
Cryptographic protocols are often implemented at upper layers of communication networks, while error-correcting codes are employed at the physical layer. In this paper, we consider utilizing readily-available physical layer functions, such as encoders and decoders, together with shared keys to provide a \textit{threshold-type} security scheme. To this end, we first consider a scenario where the effect of the physical layer is omitted and all the channels between the involved parties are assumed to be noiseless. We introduce a model for threshold-secure coding, where the legitimate parties communicate using a shared key such that an eavesdropper does not get any information, in an information-theoretic sense, about the key as well as about any subset of the input symbols of size up to a certain threshold. Then, a framework is provided for constructing threshold-secure codes from linear block codes while characterizing the requirements to satisfy the reliability and security conditions. Moreover, we propose a threshold-secure coding scheme, based on Reed-Muller (RM) codes, that meets security and reliability conditions. Furthermore, it is shown that the encoder and the decoder of the scheme can be implemented efficiently with quasi-linear time complexity. In particular, a successive cancellation decoder is shown for the RM-based coding scheme. Then we extend the setup to the scenario where the channel between the legitimate parties is no longer noiseless. The reliability condition for noisy channels is then modified accordingly, and a method is described to construct codes attaining threshold security as well as desired reliability, i.e., robustness against the channel noise. Moreover, we propose a coding scheme based on RM codes for threshold security and robustness designed for binary erasure channels along with a unified successive cancellation decoder. The proposed threshold-secure coding schemes are flexible and can be adapted for different key lengths.
\end{abstract}


\section{Introduction}
\label{section:one}

Conventional cryptosystems are often designed to be computationally secure by relying on unproven assumptions of hardness of mathematical problems. Information-theoretic security methods provide an alternative approach by constructing codes for keyless secure communication, as in wiretap channels introduced in a seminal work by Wyner \cite{wyner1975wire}. Since then, various types of wiretap channels have been considered in the literature \cite{gopala2007secrecy,oggier2011secrecy}, and with employing different coding schemes as in \cite{thangaraj2007applications,mahdavifar2011achieving}. 

Several approaches to provide security in the physical layer assuming shared secret keys have been considered in the literature. 
Such shared keys can be either fixed prior to communication as in classical cryptographic protocols or can be extracted from a source of common randomness \cite{maurer1993secret} such as characteristics of the physical layer channel, see, e.g., \cite{mathur2008radio, aldaghri2020physical, coupling1}. For instance, a variation of the wiretap channel model, where a shared secret key is assumed to be constantly generated by the legitimate parties, namely Alice and Bob, is studied in \cite{kang2010wiretap}. Another approach is to design an encryption scheme that utilizes properties of certain modulation schemes such as orthogonal frequency-division multiplexing (OFDM) to ensure security, see, e.g., \cite{zhang2014chaos,huo2014new,zhang2017design}. Other related works include using channel reciprocity properties \cite{tahir2010wireless}, classical stream ciphers at the physical layer \cite{zuquete2008physical}, introducing artificial noise \cite{zhou2009physical}, multiple-input and multiple-output (MIMO) systems \cite{goel2008guaranteeing}, public-key based McEliece cryptosystem \cite{mceliece1978public}, and using error-correcting codes for encryption \cite{mathur2006high,kim2014secure}. These prior works either consider noisy channels as in the wiretap channel model or utilize cryptographic primitives being evaluated using cryptographic measures rather than information-theoretical measures to establish security. 

Another related line of research is secure network coding, where a wiretapper has access to a certain number of edges in a network over which a source wishes to communicate messages securely. Several works have considered information-theoretic security measures while designing network codes, see, e.g., \cite{cai2002secure,bhattad2005weakly}. A similar line of work has appeared in the context of index coding, where multiple users have partial information about a set of messages and want to receive certain other messages from a central node. The eavesdropper in this scenario is then assumed to have access to a certain number of messages and a certain number of transmissions while the security of the entire message block is considered, see, e.g., \cite{dau2012security,ong2016secure}. Also, in the context of distributed storage, security guarantees are studied while having trusted storage nodes in untrusted networks. More specifically, scenarios are considered where an eavesdropper/untrusted node has access to a certain number of coded symbols and the goal is to ensure that it is not feasible to reconstruct any individual symbol of the message, e.g., the message intended for another node, see, e.g., \cite{oliveira2012coding,chung2015cyrus,soleymani2020distributed}. These prior works differ from the setting considered in this paper in two major aspects. Firstly, they are concerned with keyless techniques with information-theoretic guarantees, e.g., secret sharing, and secondly, the eavesdropper is often assumed to have access to partial information about the message/set of messages rather than the entire information block.

Utilizing error-correcting codes to provide security in the physical layer enables sharing hardware resources between reliability and security schemes in low-cost devices. Consequently, this leads to a promising approach for low-complexity applications, such as Internet-of-Things (IoT) networks. In this paper, we consider using block codes to provide a \textit{threshold-type} security scheme. A fixed key is assumed to be securely shared between the legitimate parties Alice and Bob a priori. First, we consider a scenario where the effect of the physical layer is abstracted out and all the channels between the involved parties are assumed to be noiseless. In other words, Alice communicates to Bob over a noiseless channel and her transmissions reach an eavesdropper, namely Eve, also through a noiseless channel, as shown in Figure\,\ref{fig:systemmodel}. The security condition in this model is described as follows. Alice encodes her message using the shared key while ensuring that Eve does not obtain any information about the key as well as about any subset of the input message symbols of size up to a certain threshold $t$. This condition is referred to as the $t$-threshold security condition. Then we consider the case where Alice and Bob share a noisy channel, while the eavesdropper Eve acquires Alice's transmission noise-free. The considered threshold-type security becomes relevant in applications where the knowledge of most, if not all, of the individual data symbols is needed in order to deduce meaningful knowledge about the content of the message. Examples of this type of data include measurement numbers, network commands, the index of elements in a dataset, randomly assigned identification numbers, as well as barcodes or data in any application where the data symbols are already scrambled, hashed, or masked prior to being encoded. A more detailed explanation of such applications is discussed in Section\,\ref{sec:app}. Furthermore, ensuring the security of the key in the model guarantees that it can be, theoretically, used infinitely many times without leaking any information about it or the messages to Eve.

 \begin{figure}
   \centering
   \includegraphics[trim=16.39pt 290.25pt 119.7pt 302.66pt, clip,width=3.4in]{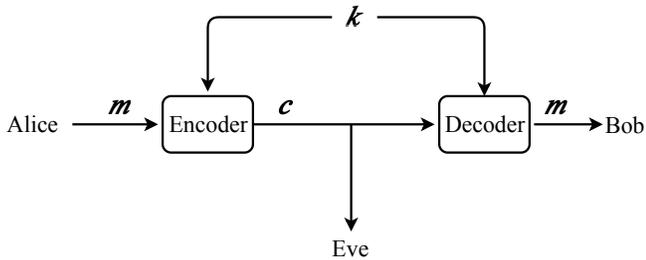}
   \caption{System setup for the proposed coding scheme.}
   \label{fig:systemmodel}
 \end{figure}
 
In the setups considered in this paper, we deviate from conventional physical-layer security settings by removing any condition on the channel from Alice to Eve; in fact, we assume this channel is noiseless. However, we still describe the schemes in a communication setting with the aim of integrating such schemes with channel coding in the physical layer. To this end, a general scheme for noiseless channels using linear block codes for the $t$-threshold-secure coding scheme is shown. Furthermore, we describe a specific construction based on RM codes \cite{reed} that meets the threshold security condition, and show an encoder and a decoder, with quasi-linear complexity, to reliably retrieve the message using the shared key. Moreover, we discuss a general method for constructing codes, closely related to concatenated codes \cite{forney1965concatenated}, for noisy channels that satisfy the threshold security requirements with respect to Eve and provide robust communication for Bob in the presence of channel noise. Also, we propose an explicit RM-based construction that is both $t$-threshold-secure and capable of correcting erasures, together with a unified successive cancellation decoder that corrects erasures and retrieves the message simultaneously given the shared key.

The rest of this paper is organized as follows. In Section \ref{sysmodel} we describe the setup and formulate the reliability and the security conditions for noiseless channels based on information-theoretic measures. The proposed coding scheme based on linear block codes is described in detail and its security and reliability are evaluated in Section \ref{coding}. Then, we describe an explicit coding scheme based on RM codes together with an encoder and a successive cancellation decoder in Section \ref{construction}. A general construction of threshold-secure codes for noisy channels together with an explicit low-complexity RM-based coding scheme for binary erasure channels (BEC) are discussed in Section \ref{robustness}. Finally, we conclude the paper in Section \ref{conclusion}, and discuss several directions for future work.

\section{System Model and Applications}\label{sysmodel}
In this section, we discuss the system model considered in this paper followed by extending certain applications of this model, as discussed in Section\,\ref{section:one}. 

\subsection{System Model}\label{sysmodel_1}
Consider a system model where Alice wishes to securely communicate with Bob, both are legitimate parties, through a noiseless channel. The eavesdropper, namely Eve, is tapping into that channel and observes all the transmitted symbols, as shown in Figure \ref{fig:systemmodel}. Alice and Bob share a common key sequence $\bk$ of length $k$, that can be used for encoding and decoding of message $\bm$ of length $m$. Both the key and the message symbols are from an alphabet of size $q$, where $q$ is a prime power. A certain known permutation $\pi(.)$ of Alice's message sequence $\bm$ together with the key sequence $\bk$ is fed as the input to the encoder, denoted by $\bu$, i.e., $\bu=\pi(\bk,\bm)$. The length of $\bu$ is $n=m+k$ and is encoded to a codeword $\bc$ of length $m$. The entries in $\bk$ as well as in $\bm$ are assumed to be independent and uniformly distributed. Alice then transmits the codeword $\bc$ to Bob over the noiseless channel. Bob receives the codeword and decodes it using the key $\bk$ to retrieve the message $\bm$. Eve observes $\bc$ and aims at extracting information about the message $\bm$ as well as the key $\bk$. In this setup, Alice and Bob agree on the encoder and the decoder a priori, which are also publicly known to Eve.

In this model, the security condition is the following. Although parts of input $\bu$ are disclosed to Eve, no knowledge, in an information-theoretic sense, about any subset of size up to a certain threshold parameter $t$ of the input symbols will be leaked to Eve. Note that this is different from the traditional measure of information-theoretic security where the mutual information between the entire message block and Eve's observation needs to be zero/almost zero. In a sense, we consider a \textit{sub-block-wise} measure of information-theoretic security. We aim at designing an encoder and a decoder for a noiseless channel that utilizes a shared key $\bk$ to encode a message $\bm$ such that the following conditions are met:

\begin{enumerate}

\item \textit{Reliability}: Bob is able to decode the message, knowing the key, with probability one, i.e.,
\begin{align}\label{reliability1}
    H(\bm|\bc,\bk)&=0.
\end{align}

\item \textit{Key security}: the codeword $\bc$ does not reveal any information about the key $\bk$, i.e.,
\begin{align}\label{security1}
    I(\bk;\bc) = 0.
\end{align}

\item \textit{$t$-threshold security}: for any $\bv \subseteq \{u_1,\!u_2,\!\mydots,u_n\}$ with $|\bv| \leq t$, we have
\begin{align}\label{security2}
    H(\bv | \bc )=H(\bv),
\end{align}

where $t$ is a design parameter specified later.

\end{enumerate}

\noindent \textbf{Remark 1.} Note that the secrecy capacity of the communication system in Figure 1, even with a relaxed security condition of $\lim_{m \to \infty}\frac{1}{m}I(\bm,\bc) = 0$, since $\bc$ is of length $m$, i.e., \textit{weak} security, is zero \cite{wyner1975wire}. In a related work \cite{kang2010wiretap}, a source of common randomness is required to generate a key with a certain rate $R_k$ to ensure non-zero secrecy capacity. However, here, a key of a fixed length is used repeatedly. In a sense, this implies that the key rate is zero as the message length grows large. 

A formal definition of a $t$-threshold secure code is defined next.

\begin{definition}
We say a code is $t$-threshold secure if it meets the reliability and security conditions, where $t$ is the maximum cardinality of any $\bv \subseteq  \{u_1,\!u_2,\!\mydots,u_n\}$ that satisfies \eqref{security2}.
\end{definition}

It is worth noting that the model considered in this paper subsumes a range of previously studied models, e.g., the perfectly-secure one-time-pad (OTP) encryption which is a code with threshold $t=m$ used once and hence, we have $H(\bm|\bc)=H(\bm)$. Another related line of work is on certain types of \textit{keyless} security schemes known as  unconditionally-secure all-or-nothing transforms (AONT) \cite{stinson2001something}. More specifically, cases are studied where the eavesdropper observes a vector $\bz$ whose elements are a subset of size $m-t$ of the set of elements of $\bc$, where $\bc$ is of length $m$ \cite{esfahani2017some}. The security condition is then translated to $H(\bv | \bz )=H(\bv)$ for all $\bv$ of size $t$ as in \cite{esfahani2017some}.

\subsection{Applications}\label{sec:app}
As briefly discussed in Section\,\ref{section:one}, the considered threshold-type security becomes relevant in applications where the entire message or significant portion of it is needed in order for an eavesdropper to obtain meaningful knowledge about the content of the message. In this section, we briefly expand on one of the applications for the described threshold security setup in Section\,\ref{sysmodel_1}. 

Consider an authentication system based on users' biometric information, such as fingerprints, e.g., as described in \cite{sutcu2005secure}, where the data is assumed to be hashed prior to encoding. Let us denote the fingerprint measurement vector as $\tilde{\bx}$. Also, let us have the following two functions: a feature extraction function $f(.)$ and a secure hash function $g(.)$. The function $f(.)$ is an arbitrary function that maps the input vector $\tilde{\bx}$ to another vector $\bx$. The hash function $g(.)$ is a mapping from an input space of size $a$ to a hash table of size $b$ with the following property:
\begin{align}
    \mathrm{Pr} (g(\bx_1) = g(\bx_2) | \bx_1 \neq \bx_2 ) = \frac{1}{b}, \label{hash_unifrom}
\end{align}
where $\bx_1$ and $\bx_2$ are any input vectors, and the resulting load factor of this hash function is $\beta=\frac{a}{b}$ \cite{cormen2009introduction}. In this example, when a user scans their fingerprint, the measurement vector $\tilde{\bx}$ is processed using $f(.)$ to produce the vector $\bx$ that is hashed using the hash function $g(.)$ to produce the hashed vector denoted as $\bm$, i.e.,
\begin{align}
    \bm=g(\bx)=g(f(\tilde{\bx})).
\end{align}
Then the hashed vector $\bm$ is the input to the threshold-secure encoder together with the key. This hashed vector is uniformly distributed by the assumption on the hash function $g(.)$ in \eqref{hash_unifrom}. The hashed vector is to be sent to a database that contains the hashed vectors of all authorized users for authentication. For an eavesdropper that aims to learn the vector $\bx$, knowledge of the entire $\bm$ is needed. Let us assume that the eavesdropper has access to the hash function $g(.)$. If $\bm$ is sent as is, the probability of successfully acquiring $\bx$ by the eavesdropper is $\frac{1}{\beta}$ since the eavesdropper can discard any vector that does not hash to the observed $\bm$. However, when using threshold-secure coding with threshold $t$, and assuming an alphabet of size $q$, this probability becomes at most $\frac{1}{\beta q^{t}}$ which is exponentially decaying with $t$. This is because the eavesdropper needs to retrieve the hashed vector $\bm$ first. Choosing an appropriate parameter $t$, e.g., in the order of a few tens, combined with the uniformity of the hash functions, is sufficient to cripple the eavesdropper in a practical setting.

\section{Coding Schemes}\label{coding}

With a slight abuse of terminology, we refer to a scheme meeting the reliability and security conditions, as described in Section\,\ref{sysmodel}, simply as a coding scheme. The coding scheme is revealed to all parties, i.e., Alice, Bob, and Eve. When constructing the coding scheme, we aim at designing an encoder and a decoder as well as specifying the code. For an input $\bu=\pi(\bk,\bm)$ the encoder produces a codeword $\bc$ as follows
\begin{align}\label{codeword0}
    \bc=\bu \textbf{W}=\pi(\bk,\bm)\textbf{W},
\end{align}
where $\textbf{W}$ is an $n\times m$ matrix with $n=m+k$. In this proposed scheme, we consider this matrix as the transpose of a generator matrix $\textbf{G}$ of a linear block code.

Consider a $[n,m,d_{\mathrm{min}}]_q$ linear block code with generator matrix $\textbf{G}$, i.e., a linear block code whose elements are from an alphabet of size $q$, and has rate $R=m/n$ and minimum distance $d_{\mathrm{min}}$. Note that in this setup no redundancy in the codeword is required since the channel is noiseless. We aim at utilizing the generator matrix $\textbf{G}$ of certain linear block codes to construct a matrix $\textbf{W}$ for our coding scheme such that the reliability and security conditions are met. 

One can assume that the length of the key is less than the length of the message; otherwise, if $k\geq m$, then the straightforward perfectly-secure one-time pad meets the conditions for $t=m$. To encode a message $\bm$, let us denote the set of indices of the rows of $\textbf{W}$ that correspond to the message symbols as $\mathcal{A} \subseteq [m+k] \,\deff\, \{1,2,\mydots,m+k\}$. Then the set of indices of the rows corresponding to the key symbols is $\mathcal{A}^c=[m+k]\setminus\mathcal{A}$. The matrix $\textbf{W}_{\mathcal{A}}$ denotes the submatrix of $\textbf{W}$ with rows indexed by $\mathcal{A}$, and the matrix $\textbf{W}_{\mathcal{A}^c}$ denotes the submatrix of $\textbf{W}$ with rows indexed by $\mathcal{A}^c$. The codeword $\bc$ is then expressed as follows:
\begin{align}\label{codeword1}
    \bc = \bm \textbf{W}_{\mathcal{A}} + \bk \textbf{W}_{\mathcal{A}^c}.
\end{align}
The choice of $\pi(.)$, which corresponds to the choice of $\mathcal{A}$ and $\mathcal{A}^c$, is critical in ensuring security and reliability conditions. Hence, we have the following definition.

\begin{definition} \label{def_proper}
A code, as described above, is called \textit{proper} if its matrix satisfies the following requirements: 
\begin{enumerate}
    \item The resulting submatrix $\textbf{W}_{\mathcal{A}}$ is full row rank, i.e.,  $\rank (\textbf{W}_{\mathcal{A}})=m$.
    \item The resulting submatrix $\textbf{W}_{\mathcal{A}^c}$ is also full row rank, i.e., $\rank (\textbf{W}_{\mathcal{A}^c})=k$.
\end{enumerate}
\end{definition}

One example of codes that are not \textit{proper} is the turbo code \cite{turbo} whose generator matrix can be written in the form $\textbf{G}=[\textbf{I}_m \hspace{0.05in} \textbf{A}_1 \hspace{0.05in} \textbf{A}_2]$ where $\textbf{I}_m$ is the identity matrix whose columns are dedicated to the message while the rest are dedicated to the key. Note that $\textbf{A}_2$ is some row-permuted version of $\textbf{A}_1$, and such a permutation may not necessarily result in $[\textbf{A}_1 \hspace{0.05in} \textbf{A}_2]^T$ being a full row-rank matrix. Hence, this code is not necessarily \textit{proper}. A code that is not \textit{proper} will result in a lower equivocation rate for Eve about the message, and leads to leakage of information about the key to Eve, as will be clarified throughout this section.

Next, we show that if a code is \textit{proper}, then it meets the reliability condition, as specified in \eq{reliability1}, and the security conditions, as specified in \eq{security1} and \eq{security2}. The following lemma shows that the reliability condition is satisfied.

\begin{lemma}\label{decodability}
Suppose that the code used in the coding scheme is \textit{proper}, as defined in Definition\,\ref{def_proper}. Then Bob can recover the message with probability one under maximum a posteriori (MAP) decoding. In other words,
\begin{align}
    H(\bm|\bc,\bk)=0.
\end{align}
\end{lemma}
\begin{proof}
By using \eqref{codeword1}, it can be observed that since Bob has $\bc$ and $\bk$ and since $\textbf{W}_{\mathcal{A}}$ is full rank, then Bob can subtract $\bk\textbf{W}_{\mathcal{A}^c}$ from $\bc$ and then find $\bm$ from $\textbf{W}_{\mathcal{A}}$, which has a unique solution.
\end{proof}

In the next theorem, we show that a \textit{proper} code meets the key security condition, as specified in \eq{security1}. Note that satisfying this condition is very critical as even a very small leakage of the key $\bk$ can lead to the entire key being revealed to Eve after using the scheme several times, thereby compromising the security of the message. 

\begin{theorem}\label{securitythm}
Suppose that the code used in the coding scheme is \textit{proper}, as defined in Definition\,\ref{def_proper}. Then the codeword $\bc$ leaks no information about the key $\bk$, i.e., 
\begin{align}
    I(\bk;\bc)=0.
\end{align}
\end{theorem}
\begin{proof} The proof is by observing the following set of equalities:
\begin{align}
    I(\bk;\bc)&=H(\bc)-H(\bc|\bk),\\
    \label{eq11}
    &=m\log_2(q)-\!H(\bm\textbf{W}_{\mathcal{A}}+\bk\textbf{W}_{\mathcal{A}^c}|\bk),\\
    \label{eq12}
    &=m\log_2(q)-\!H(\bm \textbf{W}_{\mathcal{A}}),\\
    \label{eq13}
    &= \log_2(q) (m-\rank(\textbf{W}_{\mathcal{A}})),\\
    \label{eq14}
    &= 0,
\end{align}
where \eqref{eq11} holds by \eqref{codeword1} and the uniformity of the key and message symbols, hence the codewords are uniform, \eqref{eq12} holds because $\bm$ and $\bk$ are independent, \eqref{eq13} is by noting that elements of $\bm$ are uniformly distributed and independent, and \eqref{eq14} holds because $\rank (\textbf{W}_{\mathcal{A}})=m$ as the code is \textit{proper} according to Definition \ref{def_proper}.
\end{proof}

Additionally, to fully justify the reuse of $\bk$ for multiple encodings, we include the following corollary.

\begin{corollary}
    Suppose that the code used in the coding scheme is \textit{proper}, as defined in Definition\,\ref{def_proper}. Then the codewords $(\bc_1,\bc_2,\mydots,\bc_v)$ of the independent and uniform messages $(\bm_1,\bm_2,\mydots,\bm_v)$ leak no information about the key $\bk$, i.e., 
    \begin{align}
        I(\bk;\bc_1,\bc_2,\mydots,\bc_v)=0.
    \end{align}
\end{corollary}

\begin{proof} The proof is by observing the following set of equalities:
\begin{align}
    I(\bk;\bc_1,\bc_2,\mydots,\bc_v)&=H(\bc_1,\bc_2,\mydots,\bc_v)\nonumber\\
    &\hspace{0.085 in}-H(\bc_1,\bc_2,\mydots,\bc_v|\bk),\\
    \label{eq51}
    &=vm\log_2(q)\nonumber\\
    &\hspace{0.085 in}-H(\bm_1\textbf{W}_{\mathcal{A}},\bm_2\textbf{W}_{\mathcal{A}},\mydots,\bm_v\textbf{W}_{\mathcal{A}}),\\
    \label{eq52}
    &=vm\log_2(q)-vH(\bm_i \textbf{W}_{\mathcal{A}}),\\
    \label{eq53}
    &=v \log_2(q) (m-\rank(\textbf{W}_{\mathcal{A}})),\\
    \label{eq54}
    &= 0,
\end{align}
where \eqref{eq51} holds by \eqref{codeword1}, the uniformity of codewords, and the independence of the key and messages, \eqref{eq52} holds by independence and uniformity of messages $(\bm_1,\bm_2,\mydots,\bm_v)$, where $\bm_i$ is uniformly distributed, \eqref{eq53} is by noting that elements of message $\bm_i$ are uniformly distributed and independent, and \eqref{eq54} holds because $\rank (\textbf{W}_{\mathcal{A}})=m$ as the code is \textit{proper} as in Definition \ref{def_proper}.
\end{proof}

The following lemma is well-known. However, it is included here as it is instrumental in characterizing the threshold security of coding schemes based on linear block codes. 

\begin{lemma} \cite{roth2006introduction}\label{min_dis}
For an $[n,m,d_{\mathrm{min}}]_q$ linear block code with generator matrix $\textbf{G}$, any submatrix of $\textbf{G}$ of size $m \times (n-|\mathcal{D}|)$ obtained by deleting columns indexed by elements of $\mathcal{D}$, where $\mathcal{D} \subseteq [n]$ with $|\mathcal{D}|=d_{\mathrm{min}}-1$, has full row rank, i.e.,
\begin{align}
    \rank(\textbf{G}_{\mathcal{D}^c})=m.
\end{align}
\end{lemma}

In the next theorem, we characterize the threshold security of coding schemes based on linear block codes.

\begin{theorem}\label{sub_block_sec}
A coding scheme constructed by a matrix $\textbf{W} = \textbf{G}^{\text{T}}$, where $\textbf{G}$ is the generator matrix of an $[n,m,d_{\mathrm{min}}]_q$ linear block code, is $t$-threshold secure, where $t=d_{\mathrm{min}}-1$, i.e., we have
\begin{align}
    H(\bv|\bc)=H(\bv),
\end{align}
for any $\bv \subseteq \{u_1,u_2,\mydots,u_n\}$ with $|\bv|=t$, and $t$ is the maximum value for which this condition holds.
\end{theorem}

\begin{proof}
Let $\bu$ denote the input to the encoder for the coding scheme, as specified in \eq{codeword0}. Suppose that $\bv$ consists of elements of $\bu$ indexed by $\mathcal{B}=\{i_1,i_2,\mydots,i_t\} \subseteq [n]$, and $\tilde{\bu}$ consists of elements of $\bu$ indexed by $\mathcal{B}^c=[n]\setminus\mathcal{B}$. Then we have the following:
\begin{align}
    I(\bv ; \bc)&=H(\bc)-H(\bc|\bv),\label{eq32}\\
    &=m\log_2(q)-H(\tilde{\bu}\textbf{W}_{\mathcal{B}^c}+\bv \textbf{W}_{\mathcal{B}}|\bv),\label{eq33}\\
    &=m \log_2(q) -H(\tilde{\bu}\textbf{W}_{\mathcal{B}^c}),\label{eq34}\\
    &=\log_2(q)(m-\rank(\textbf{W}_{\mathcal{B}^c})),\label{eq35}\\
    &=0,\label{eq36}
\end{align}
where \eqref{eq33} follows due to codewords being uniformly distributed and expansion of random variables, \eqref{eq34} holds by the independence of $\bv$ and $\tilde{\bu}$, \eqref{eq35} holds due to the uniformity of $\tilde{\bu}$, and \eqref{eq36} holds by Lemma \ref{min_dis} with $t=d_{\mathrm{min}}-1$. Since the mutual information $I(\bv ; \bc)$ is zero, it implies that the $t$-threshold security criteria is met for the parameter $t=d_{\mathrm{min}}-1$, i.e.,
\begin{align}
    H(\bv|\bc)=H(\bv),
\end{align}
for any $\bv$ with $|\bv|=t$, where $t=d_{\mathrm{min}}-1$.

Next, we need to show that $t=d_{\mathrm{min}}-1$ is the maximum value for which the threshold security condition holds. Consider a codeword in the codebook generated by $\textbf{G}$ that has the Hamming weight equal to $t+1 = d_{\mathrm{min}}$ with non-zero elements at indices denoted by $\mathcal{F}=\{i_1,i_2,\mydots,i_{t+1}\}$. Then we have the following:
\begin{align}
    H(u_{i_1},\mydots,u_{i_{t+1}}|\bc)&= H(u_{i_1},\mydots,u_{i_{t}}|\bc)\nonumber\\
    &\hspace{0.085 in}+H(u_{i_{t+1}}|\bc, u_{i_1},\mydots,u_{i_{t}}),\label{eq41}\\
    &=H(u_{i_1},\mydots,u_{i_{t}}|\bc),\label{eq42}\\
    &\neq H(u_{i_1},\mydots,u_{i_{t+1}}),\label{eq43}
\end{align}
where \eqref{eq41} follows from the chain rule of entropy, and \eqref{eq42} holds because there exists a linear combination of the entries of $\bc=(c_1,c_2,\mydots,c_m)$ such that $\sum_{i=1}^{m} \lambda_i c_i = \sum_{j \in \mathcal{F}} \gamma_j u_j$. Hence, the second term becomes zero, since $u_{i_{t+1}}$ is uniquely determined given $\bc$ and $\{u_{i_1},\mydots,u_{i_{t}}\}$. Therefore, due to \eqref{eq43}, the threshold security condition does not hold for $t+1 = d_{\mathrm{min}}$.
\end{proof}

\begin{corollary}
\label{Sing}
For any $t$-threshold secure coding scheme, constructed from a linear block code, with message length $m$, key length $k$, and code length $n = m+k$, we have $t \leq k$.  
\end{corollary}
\begin{proof}
The proof follows by \Tref{sub_block_sec} together with Singleton bound on the minimum distance of a code. 
\end{proof}

Next, we characterize Eve's equivocation about the entire message $\bm$ after observing the codeword.
\begin{corollary}\label{evesequivocation}
If the code is \textit{proper}, then Eve's equivocation about the entire encoded message $\bm$ after observing the codeword is equal to the entropy of the key, i.e.,
\begin{align}
    H(\bm|\bc)=k\log_2(q).
\end{align}
\end{corollary}
\begin{proof}
We have the following
\begin{align}
    H(\bm|\bc)&=H(\bm)-H(\bc)+H(\bc|\bm),\label{eq21}\\
    &=\!H(\bk \textbf{W}_{\mathcal{A}^c} + \bm \textbf{W}_{\mathcal{A}}|\bm),\label{eq23}\\
   &=H(\bk\textbf{W}_{\mathcal{A}^c}),\!\label{eq24}\\
    &=k\log_2(q), \label{eq25}
\end{align}
where 
\eqref{eq23} follows due to the uniformity of messages and codewords, and expansion of random vectors, \eqref{eq24} holds because of the independence of $\bm$ and $\bk$, and \eqref{eq25} holds by noting that the matrix $\textbf{W}_{\mathcal{A}^c}$ is full row rank since the code is \textit{proper}.
\end{proof}

The statement of \Cref{evesequivocation} can be also rephrased by stating that the probability of successfully retrieving the entire message block by Eve is equal to $q^{-k}$.

Now that we have established the properties that the coding schemes based on linear block codes satisfy, we need to show how to maximize the threshold $t$ as stated in \Cref{Sing}, provided that $q$ is large enough. To this end, we utilize maximum distance separable (MDS) codes to arrive at the following theorem. 

\begin{theorem}\label{exitence}
For any message length $m$ and key length $k$, there exists a \textit{proper} code with threshold $t = k$, provided that the alphabet size $q \geq m+k+1$. 
\end{theorem}

\begin{proof}
To prove the theorem, we give an example of a code that is shown to be \textit{proper} with $t=k$. We utilize Reed-Solomon (RS) codes, which are a well-known family of codes that are maximum distance separable (MDS) codes, i.e., $d_{\mathrm{min}}=n-m+1=k+1$ \cite{roth2006introduction}. For any $[n,m,d_{\mathrm{min}}]_q$ RS code, all we need to show is that the matrix $\textbf{W}$ which is the transpose of the generator matrix $\textbf{G}$ of the RS code can be used to construct a \textit{proper} code. One of the properties of MDS codes is that every set of $m$ columns of the matrix $\textbf{G}$ are linearly independent \cite[Proposition 11.4]{roth2006introduction}. Note that rows of $\textbf{W}$ correspond to columns of $\textbf{G}$. Hence, any choice of $m$ columns of $\textbf{G}$ will have rank $m$, and the remaining $k$ columns of $\textbf{G}$ will also have rank $k$ as it is assumed that $k<m$. Therefore, the code generated by $\textbf{W}$ is \textit{proper}, with threshold $t=k$.
\end{proof}

Note that the straightforward Gaussian elimination method, with complexity $O(m^3)$, can be always used for decoding of coding schemes based on linear block codes. However, when the underlying linear block code belongs to well-known families of linear block codes, e.g., Reed-Solomon codes, it is desirable to study low-complexity decoders for the resulting coding schemes using the off-the-shelf encoding/decoding methods. For instance, low-complexity decoding of RS codes is based on a low-complexity computation of the inverse of a Vandermonde matrix. Now, for the coding schemes based on RS codes, the evaluation points for the RS encoder are chosen as consecutive powers of $\alpha$, where $\alpha$ is a primitive element of $\Fq$. The specific choice of the message and key indices is as follows: the first $m$ rows of $\textbf{W}$ are dedicated for the message $\bm$, and the last $k$ rows of $\textbf{W}$ are dedicated for the key $\bk$. Since $\textbf{W}$ is a Vandermonde matrix, this choice of message indices together with the specific choice of evaluation points result in a scenario where the submatrix $\textbf{W}_{\mathcal{A}}$ is also a Vandermonde matrix. To decode a codeword using the key, the decoder computes $\bm = (\bc - \bk \textbf{W}_{\mathcal{A}^c}) \textbf{W}_{\mathcal{A}}^{-1}$. Note that the inverse of a square Vandermonde matrix of order $m$ can be computed with complexity $\mathcal{O}(m^2)$ \cite{bjorck1970solution}. This results in $\mathcal{O}(m^2)$ complexity for the decoding in coding schemes based on RS codes.

\section{Low-Complexity Construction}\label{construction}
In this section, we focus on designing binary codes to meet the reliability and security conditions while providing encoding and decoding algorithms with linear/quasi-linear complexity. To this end, we consider Reed-Muller codes due to their recursive construction and low-complexity decoder. In addition, since they are designed with the objective of maximizing the minimum distance, given their particular recursive structure, we can achieve a reasonably high threshold $t$ for the $t$-threshold security.

It is worth noting that various types of decoders for Reed-Muller codes are proposed in the literature, see, e.g., \cite{reed,schnabl1995soft,dumer2004recursive}. However, the proposed decoder here differs from these works as it has different constraints and objectives. The goal of the decoder here is not to correct errors, but rather to successfully recover the message from an error-free codeword encoded by having the message as well as the key as the input. Also, the message cannot be retrieved completely without complete knowledge of the key itself. This shows the need to adapt or modify encoders/decoders in such a way that they can be utilized for threshold-security decoding accordingly.

\subsection{Encoder}\label{encoder}

First, a brief description of Reed-Muller codes is provided. An RM($s, r$) code is a $[2^s,\sum_{i=0}^{r} \binom{s}{i}, 2^{s-r}]_2$ linear block code. The generator matrix of the RM($s, r$) code, denoted by $\textbf{G}(s,r)$, is obtained by keeping the rows with the Hamming weight of at least $2^{s-r}$ from the matrix $\textbf{F}^T=(\textbf{F}_2^{\otimes s})^T$ and removing the remaining rows, where $\otimes$ denotes the Kronecker product, $T$ is the transpose operator, and $\textbf{F}_2$ is the following kernel matrix
\begin{align}
    \textbf{F}_2=\begin{bmatrix} 
    1 & 0 \\
    1 & 1 
    \end{bmatrix}.
\end{align}

Although there are different ways of describing the encoding and the generator matrix of RM codes, the above description helps us to choose the message and key indices, which is the next step towards designing a code that is \textit{proper}. Due to the recursive structure of $\textbf{F}$, it can be observed that indices of the rows with the lowest weight, the second lowest weight, etc, from $\textbf{F}$ correspond to indices of columns with the highest column weight, the second highest weight, etc, from $\textbf{F}$, respectively. When specifying the matrix $\textbf{G}(s,r)$ as a sub-matrix of $\textbf{F}^T$ we choose the set of indices of the removed rows from $\textbf{F}^T$ as $\mathcal{A}^c$ to assign the rows of $\textbf{W}$ dedicated for the key, while the indices of the remaining rows are used as the message indices $\mathcal{A}$. Then we have the following proposition.

\begin{proposition}
The choice of the sets $\mathcal{A}$, and $\mathcal{A}^c$ as mentioned above results in a \textit{proper} code.
\end{proposition}
\begin{proof}
To prove this proposition, it suffices to show that $\textbf{W}_{\mathcal{A}}$ and $\textbf{W}_{\mathcal{A}^c}$ are both full row rank.

First, it is shown that $\textbf{W}_{\mathcal{A}}$ is full row rank. Note that for a full rank lower-triangular matrix, a submatrix obtained by removing a subset of columns and rows with the same indices results also in a full rank lower-triangular matrix. Also, note that $\mathcal{A}^c$ is the subset of indices of deleted columns as well as that of the rows dedicated for the key from $\textbf{F}$. Hence, the matrix $\textbf{W}_{\mathcal{A}}$ is full row rank.

Next, we show that $\textbf{W}_{\mathcal{A}^c}$ is full row rank. This is done by induction. Note that $k < m$ is assumed, as mentioned before. Also, to simplify the proof, let us have $r'=s-r$, and also re-express $k$ and $m$ in the remainder of the proof as follows
\begin{align}
    k=\sum_{i=0}^{r'} \binom{s}{i}, \nonumber
\end{align}
and 
\begin{align}
    m=\sum_{i=r'+1}^{s} \binom{s}{i}, \nonumber
\end{align}
where we have $r'\leq \floor{\frac{s-1}{2}}$. Note that $\textbf{W}_{\mathcal{A}^c}$ contains the $\sum_{i=0}^{r'} \binom{s}{i}$ rows dedicated for the key from $\textbf{F}$ with the same number of lowest-weight columns removed. Let this matrix be also denoted by $\textbf{F}(s,r')$. Let also $\textbf{F}'(s,r')$ denote the matrix that contains the $\sum_{i=0}^{r'} \binom{s}{i}$ rows dedicated for the key from $\textbf{F}$ with only $\sum_{i=0}^{r'-1} \binom{s}{i}$ lowest weight columns removed. Due to the recursive structure of the matrix $\textbf{F}$, $\textbf{F}(s,r')$ can be expressed as follows:

\begin{align}\label{matrixfsr}
\textbf{F}(s,r')=\begin{bmatrix} 
\textbf{F}(s-1,r'-1) & \textbf{0} \\
\textbf{F}'(s-1,r') & \textbf{F}(s-1,r')
\end{bmatrix}.
\end{align}

Next, we show that the matrix $\textbf{F}(s,r')$ is full row rank for the maximum value $r'=\floor{\frac{s-1}{2}}$ and for $s \geq 2$ by induction on $s$. Then it will be discussed why this also holds for $r'<\floor{\frac{s-1}{2}}$.

\textbf{Step 1:} The induction basis is for $s=2$ and $r'=0$, and for $s=3$ and $r'=1$, which can be easily verified, i.e., for $s=2$ and $r'=0$, the rank of $\textbf{F}(2,0)$ is $1$. Also, for $s=3$ and $r'=1$, the rank of $\textbf{F}(3,1)$ is $4$.

\textbf{Step 2:} Suppose that the induction hypothesis holds for $s$ and $s$ is odd. Then we have the following matrix:
\begin{align}
\textbf{F}(s+1,r')=\begin{bmatrix} 
\textbf{F}(s,r'-1) & \textbf{0} \\
\textbf{F}'(s,r') & \textbf{F}(s,r')
\end{bmatrix}.
\end{align}
We need to show that $\rank(\textbf{F}(s+1,r'))=\sum_{i=0}^{r'} \binom{s+1}{i}$. Note that $\textbf{F}(s,r')$ is full row rank by induction hypothesis, i.e., $\rank(\textbf{F}(s,r'))=\sum_{i=0}^{r'} \binom{s}{i}$. Then $\textbf{F}(s,r'-1)$, which contains a subset of the rows in $\textbf{F}(s,r')$, is also full row rank. Hence, we have $\rank(\textbf{F}(s,r'-1))=\sum_{i=0}^{r'-1} \binom{s}{i}$. Therefore,
\begin{align}
\rank(\textbf{F}(s+1,r'))&=\rank(\textbf{F}(s,r'-1))\nonumber\\
    &\hspace{0.1 in}+\rank(\textbf{F}(s,r')),\\
&=\sum_{i=0}^{r'-1} \binom{s}{i}+\sum_{i=0}^{r'} \binom{s}{i},\\
&=\sum_{i=0}^{r'} \binom{s+1}{i},
\end{align}
which is equal to the number of rows in $\textbf{F}(s+1,r')$. Hence, it is full row rank.

For even $s$ with corresponding parameter $r'$, we need to show the following matrix is full row rank
\begin{align}\label{matrixF}
\textbf{F}(s+1,r'+1)=\begin{bmatrix} 
\textbf{F}(s,r') & \textbf{0} \\
\textbf{F}'(s,r'+1) & \textbf{F}(s,r'+1)
\end{bmatrix}.
\end{align}
First, we have $\rank(\textbf{F}(s,r'))=\sum_{i=0}^{r'} \binom{s}{i}$ by induction hypothesis. Regarding $\rank(\textbf{F}'(s,r'+1))$, we can see that $\textbf{F}'(s,r'+1)$ has $\sum_{i=0}^{r'} \binom{s}{i}$ rows that are also included in $\textbf{F}(s,r')$. However, when considering the indices of such rows in $[\textbf{F}'(s,r'+1)\hspace{0.05in} \textbf{F}(s,r'+1)]$, the corresponding rows are independent from all other rows in $[\textbf{F}(s,r'), \hspace{0.05in} \textbf{0}]$. Furthermore, there are $\binom{s}{r'+1}$ additional rows in $\textbf{F}'(s,r'+1)$ that are linearly independent from the remaining rows due to the structure of the zero blocks in this matrix, similar to \eqref{matrixfsr}. We can then find the rank of $\textbf{F}(s+1,r'+1)$ as follows
\begin{align}
\rank(\textbf{F}(s+1,r'+1))&=\rank(\textbf{F}(s,r'))\nonumber\\
    &\hspace{0.1 in}+\rank(\textbf{F}'(s,r'+1)),\\
&=\sum_{i=0}^{r'} \binom{s}{i}+\sum_{i=0}^{r'} \binom{s}{i}\nonumber\\
    &\hspace{0.1 in}+\binom{s}{r'+1},\\
&=\sum_{i=0}^{r'+1} \binom{s+1}{i}.
\end{align}
Hence, $\textbf{F}(s+1,r'+1)$ is full row rank, and the induction hypothesis holds for $s+1$ with the maximum value of $r'$. For keys of shorter lengths, it is straightforward to see that for any $r''<r'$, the matrix $\textbf{F}(s,r'')$ whose rows are a subset of $\textbf{F}(s,r')$ with additional columns inserted at different locations is also full row rank. This completes the proof. 
\end{proof}

\noindent \textbf{Remark 2.} In the proposed scheme based on RM codes, we have $n=2^s$, $m = \sum_{i=0}^{r} \binom{s}{i}$, for some $r\geq s/2$, and $k = n-m < m$. Note that the underlying RM code has rate $R>\frac{1}{2}$. By using Theorem \ref{sub_block_sec} and noting that the minimum distance of the underlying code is $2^{s-r}$, the achievable threshold security parameter $t$ for the RM-based scheme with parameters $(s,r)$ is $t=2^{s-r}-1$. Note that, in general, for an RM code of constant rate, i.e., $R = O(1)$, we have $r = s/2+O(\sqrt{s})$. Hence, the threshold security parameter of the corresponding scheme is $t = \sqrt{n} \exp(O(\sqrt{\log n}))$.

\subsection{Decoder}

In this part, we discuss a low-complexity successive cancellation (SC) decoder to decode the message in the RM-based coding scheme while utilizing the shared key. As Reed-Muller codes are closely related to polar codes \cite{arikan}, a decoder closely related to that of polar codes described in \cite{arikan} is natural. However, there are fundamental differences that will be clarified throughout this section. 

The decoder is described in Algorithm \ref{decoderrecursive}. We first embed erasures within the entries of the codeword $\bc$ in order to get a vector of length $n$, denoted by $\bz$, by inserting the erasures at locations indexed by $\mathcal{A}^c$. More specifically, $\bz=\pi_1(\be_k,\bc)$ where $\bc$ is the codeword and $\be_k$ is an erasure vector of length $k$ such that the permutation places the erasures at locations denoted by $\mathcal{A}^c$. Note that, as mentioned before, $\mathcal{A}^c$ corresponds to the location of the key bits at the encoder. 

The decoder takes the key bits $\bk$, the codeword embedded with erasures $\bz=\pi_1(\be_k,\bc)$, indices of the key bits $\mathcal{A}^c$ and a recursion index $i$ as inputs, and outputs the vector $\bu=[u_1,u_2,\mydots,u_n]=\pi(\bk,\bm)$ from which the message can be retrieved $\bm=\bu_{\mathcal{A}}$. The high-level idea of the decoder is as follows. The vector $\bz$ is divided into two parts; $z_{1}^{n/2}=[z_1,z_2,\mydots,z_{n/2}]$ and $z_{n/2+1}^{n}=[z_{n/2+1},z_{n/2+2},\mydots,z_{n}]$, that are decoded successively. As opposed to the SC decoder of polar codes \cite{arikan}, the second sub-block is processed first, cancelled from the first sub-block, and then the first sub-block is processed. Each of these sub-blocks is also decoded recursively by splitting them into two parts and so on. 

\noindent \textbf{Remark 3.} When describing the recursive SC decoding process we often use the binary tree terminology in which the input codeword, i.e., $\bz$, is assigned to the root of the tree and then the first and the second sub-blocks are assigned to the \textit{left child} and the \textit{right child}, respectively. The decisions are made at the leaves of the tree and then are re-encoded and propagated back through the tree, see, e.g., \cite{mahdavifar2014fast} for more details.

\begin{algorithm}[t]
\caption{Successive cancellation decoder (Decoder)}
\begin{algorithmic}[1]\label{decoderrecursive}
\STATE Initialization: $i=1$.
\STATE Input: $\bk$, $z_1^n=\pi_1(\be_k,\bc)$, $\mathcal{A}^c$, $i$.
\STATE Output: $h_1^n$, $u_1^n$.
    \IF{$n=2$}
      \IF{$z_2=e$}
        \STATE $u_i=k_i$
        \ELSE
        \STATE  $u_{i}=z_{2}$
        \ENDIF
        \IF{$z_1=e$}
            \STATE  $u_{i-1}=k_{i-1}$
            \ELSE
            \STATE $u_{i-1}=u_i \oplus z_1$
        \ENDIF
        \STATE $h_1^n=[u_{i-1}\oplus u_{i}, u_i]$
        \ELSE
            \STATE $\bh' \leftarrow$ Decoder($\bk_2,z_{n/2+1}^n,\mathcal{A}_2^c,2i$)
            \STATE $\Bar{z}_{1}^{n/2}=\bh'  \oplus z_1^{n/2}$
            \STATE $\bh'' \leftarrow$ Decoder($\bk_1,\Bar{z}_{1}^{n/2},\mathcal{A}_1^c,2i-1$)
            \STATE $h_1^n=[\bh''\oplus \bh', \bh']$
        \ENDIF
        \RETURN $h_1^n$
\end{algorithmic}
\end{algorithm}

The following claim verifies that the decoder successfully outputs the message bits with probability $1$ for any key length. Note that since the proof follows by induction, we discard the assumption that $k\leq m$ and simply show the claim for any $k \leq n$. 

\begin{claim}
The RM-based coding scheme can be successfully decoded using the SC decoder in Algorithm \ref{decoderrecursive} for any key length $k \leq n$.
\end{claim}

\begin{proof}
We use induction on $l$, where $n=2^l$, to show that the claim holds.

\textbf{Step 1:} For the induction basis, consider $n=2$. We need to show decoding is successful for $k=0,1,2$. For $k=0$, which corresponds to the case with no erasure, the induction hypothesis holds trivially as $\textbf{F}$ is non-singular. For $k=1$, one needs to show the induction hypothesis for both possible cases for $\mathcal{A}^c$. First, let us consider that $z_1=e$ and $z_2=c_1$, which corresponds to $u_1=k_1$, and $u_2=m_1$. In this case, the decoder outputs $u_1=k_1$ and $u_2=z_2$. For the other case where $z_1=c_1$ and $z_2=e$, which corresponds to $u_1=m_1$ and $u_2=k_1$, the decoder first corrects the erasure, assigning $u_2=k_1$. It then computes $u_1=m_1=u_2\oplus z_1=k_1\oplus z_1$. Finally, we show it succeeds for $k=2$, where both $z_1$ and $z_2$ are erased. Then $u_1=k_1$ and $u_2=k_2$ and the decoder is successful.

\textbf{Step 2:} Now, suppose that the induction hypothesis holds for $n=2^l$ and for any $k \leq 2^l$, where $k$ is the length of the key, regardless of the indices of the key bits. However, note that, as specified before, the row indices corresponding to the key bits and the column indices corresponding to the erasures are the same and are both denoted by $\mathcal{A}^c$. We now show that the claim is true for $n=2^{l+1}$ and any $k \leq 2^{l+1}$. Let us split the key indices $\mathcal{A}^c$ into two sets, $\mathcal{A}_1^c$ and $\mathcal{A}_2^c$, with sizes $|\mathcal{A}_1^c|=k_1$ and $|\mathcal{A}_2^c|=k_2$, where $k=k_1+k_2$, as follows. The set $\mathcal{A}_1^c$ consists of the indices of erasures in $z_{1}^{n/2}$. Also, let $\bk_1$ denote the corresponding part of the key of size $k_1$. Similarly, $\mathcal{A}_2^c$ consists of the indices of erasures in $z_{n/2+1}^{n}$. Also, let $\bk_2$ denote the corresponding part of the key of size $k_2$. First, the right child with input $z_{n/2+1}^{n}$, which has $k_2$ erasures, is processed. Note that there are also $k_2$ known key bits indexed by $\mathcal{A}_2^c$ in the second half sub-block $u_{n/2+1}^n$. Note that the decoder for the right child has an input of length $n'=2^l$ and $k'=k_2$ erasures as well as key bits $\bk_2$ indexed by $\mathcal{A}_2^c$. The decoder succeeds by the induction hypothesis. The right child then passes 
$$
u^{n}_{n/2+1}\textbf{F}_2^{\otimes l} \oplus z_{1}^{n/2}=\bh'\oplus z_{1}^{n/2}=\Bar{z}_{1}^{n/2}
$$
to the left child. The decoder is then run on $\Bar{z}_{1}^{n/2}$, which is of length $n'=2^l$ and has $k'=k_1$ erasures and key bits $\bk_1$ indexed by $\mathcal{A}_1^c$. The decoder is successful on this node as well by the induction hypothesis. Hence, the decoder is successful for $n=2^{l+1}$ which completes the proof of the claim.
\end{proof} 

\section{Robustness}\label{robustness}
In this section we study a natural scenario for extension of the considered setup and the results. In particular, it is assumed that a noisy channel is present between the legitimate parties and the goal is to study the robustness of the framework and the proposed solution when channel noise is present. 

The revised system model, shown in Figure \ref{fig:systemmodel_ch}, is as follows: the channel between Alice and Bob is no longer noiseless, and it can be a certain type of channel to be studied, e.g., binary symmetric channel (BSC), binary erasure channel (BEC), additive-white Gaussian noise channel (AWGN), etc. However, for the eavesdropper, we still consider a worst-case scenario from the legitimate parties' perspective. In other words, it is assumed that Eve receives the transmitted codeword through a noiseless channel, and hence, she has access to the codeword error-free. Alice aims to utilize a coding scheme such that the threshold security requirement at Eve is satisfied while establishing a reliable communication with Bob that is robust in the presence of channel noise. 

\begin{figure}
   \centering
   \includegraphics[trim=35.7pt 640.14pt 137.25pt 10.3pt, clip,width=3.4in]{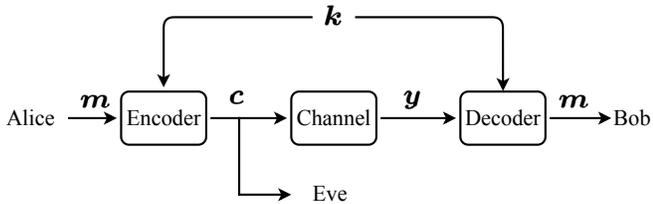}
   \caption{Modified setup for the proposed coding scheme in the presence of a noisy channel.}
   \label{fig:systemmodel_ch}
\end{figure}

Note that the assumption on Eve's observation here makes it reasonable to keep the conditions in \eqref{security1} and \eqref{security2} the same in this revised model. On the other hand, the reliability condition in \eqref{reliability1} needs to be modified to account for the noisy channel. We do this from a conventional block coding perspective where reliability is measured in terms of a certain number of errors and erasures that can be corrected. More specifically, the reliability condition is still stated as 
\begin{align}\label{reliability2}
    H(\bm|\by,\bk)&=0,
\end{align}
provided that the number of erasures and errors introduced in $\by$ satisfies a certain condition that depends on the underlying coding scheme. For instance, consider coding schemes based on linear block codes. Suppose that the minimum distance of the robustness coding scheme is $D_{\mathrm{min}}$ when the key is fixed, which is different from the minimum distance of the threshold security coding scheme, i.e., $d_{\mathrm{min}}$. Then the condition on the number of errors and erasures is simply $2 \tau+\rho \leq D_{\mathrm{min}}-1$, where $\tau$ is the number of errors and $\rho$ is the number of erasures, same as in conventional block codes. 

In the remainder of this section, we discuss a general method to construct codes for threshold security and robustness, and describe an explicit low-complexity construction based on Reed-Muller codes for binary erasure channels along with a SC decoder.

\subsection{General construction}\label{robustconstruction}

A straightforward solution to construct coding schemes for the setting described in this section is by utilizing concatenation of two codes. More specifically, a coding scheme, constructed to guarantee the desired threshold security in the error-free case, would be concatenated with an inner code, that can be an off-the-shelf block code, to guarantee the desired reliability for Alice-Bob communication. Although this solution is straightforward, one needs to ensure that the threshold security guarantee is not compromised when more redundancy is added through the inner encoder which will be then revealed to Eve. 

In the aforementioned concatenation scheme, the overall encoder and decoder at Alice and Bob, respectively, are referred to as \textit{supercoder} and \textit{superdecoder}, respectively. The construction of the concatenated scheme is described in more details next. Consider a \textit{proper} coding scheme, that guarantees threshold security requirement, that is obtained from an $[n,m,d_{\mathrm{min}}]_q$ linear block code with the generator matrix $\textbf{W}^T$. Also, consider an error-correcting code, used as an inner code to guarantee the reliability, that is an $[N,m,D_{\mathrm{min}}]_q$ linear block code with the generator matrix denoted by $\textbf{G}_r$. It is important to note that both codes have the same dimension $m$. 

The encoding process is as follows. First, $\bu=\pi(\bk,\bm)$ is passed through the outer threshold security encoder that multiplies $\bu$ by $\textbf{W}$. The result is then passed to the inner encoder, which multiplies its input by $\textbf{G}_r$. Then the resulting codeword $\bc=\bu \textbf{W} \textbf{G}_r$ is transmitted to Bob through the noisy channel. Bob receives a corrupted version of the codeword $\bc$, denoted as $\by$, and passes it through the decoder consisting of an inner decoder and an outer decoder. The inner decoder retrieves $\tilde{\bc}=\bu \textbf{W}$. Note that we have $\tilde{\bc}$ error-free provided that the number of errors and/or erasures satisfies the given condition on the reliability guarantee of the inner code. Then $\tilde{\bc}$ together with the key $\bk$ are passed through the outer decoder, designed for the threshold security coding scheme; hence, retrieving $\bm$. The following lemma states that this construction does not compromise the key and threshold security conditions.
\begin{lemma}\label{robustcode}
The aforementioned concatenation scheme results in a $t$-threshold secure code.
\end{lemma}

\begin{proof}
To show that the lemma holds, we need to have $\rank(\textbf{W}\textbf{G}_r)=m$, $\rank(\textbf{W}_{\mathcal{A}}\textbf{G}_r)=m$, $\rank(\textbf{W}_{\mathcal{A}^c}\textbf{G}_r)=k$, and $\rank(\textbf{W}_{\mathcal{B}^c}\textbf{G}_r)=m$, where $\mathcal{A}$ and $\mathcal{A}^c$ are chosen such that the code is \textit{proper}, as stated in Definition \ref{def_proper}, and $\mathcal{B}^c$ is as defined in Theorem \ref{sub_block_sec}. It can be observed that all these equations hold simply because $\textbf{G}_r$ is full row rank.
\end{proof}

\subsection{Low-complexity construction}\label{robustlowcomplexity}
In this section, we aim at presenting a unified coding scheme, for threshold security and robustness, that can be decoded using one unified SC decoder. This would potentially result in more efficient hardware implementation and improved latency compared to the general concatenated scheme. 

In particular, a scenario with binary symbol erasures is considered, where at most $\rho=D_{\mathrm{min}}-1$ erasures are assumed to occur with $D_{\mathrm{min}}$ being the minimum distance of the underlying code. For the proposed coding scheme, an encoder is presented together with a superdecoder that simultaneously corrects erasures and decodes the message using the key. To this end, the coding scheme presented for noiseless channels in Section \ref{construction} is extended to be utilized along with an RM-based code to handle binary erasures.

\subsubsection{Encoder}\label{robustencoder}

In the considered scheme, the same RM code is used for threshold security and robustness. More specifically, an RM($s, r$) is used, which is, as previously described, a $[2^s,\sum_{i=0}^{r} \binom{s}{i}, 2^{s-r}]_2$ with the generator matrix denoted by $\textbf{G}(s,r)$. The encoder with input $\bu$, consisting of both the message and the key, outputs the codeword $\bc$ specified as follows:
\begin{align}
    \bc &= \bu \textbf{G}^T(s,r) \textbf{G}(s,r)=\bu \widetilde{\textbf{G}}(s,r),\label{robustcodeword}
\end{align}
where $\widetilde{\textbf{G}}(s,r)$ is a notation introduced here to denote $\textbf{G}^T(s,r) \textbf{G}(s,r)$. Note that the encoder can be implemented recursively, since $\widetilde{\textbf{G}}(s,r)$ can be expressed recursively as shown in \eqref{combinedmatrix4}.
\begin{figure*}
\begin{align}
    \widetilde{\textbf{G}}(s,r)& = \textbf{G}^T(s,r)\textbf{G}(s,r),\label{combinedmatrix1}\\
    & =\begin{bmatrix}
    \textbf{G}^T(s\!-\!1,r\!-\!1) & \textbf{G}^T(s\!-\!1,r)\\
    \textbf{0} & \textbf{G}^T(s\!-\!1,r)
    \end{bmatrix} \begin{bmatrix}
    \textbf{G}(s\!-\!1,r\!-\!1) & \textbf{0}\\
    \textbf{G}(s\!-\!1,r) & \textbf{G}(s\!-\!1,r)
    \end{bmatrix},\label{combinedmatrix2}\\
    & = \begin{bmatrix}
    \textbf{G}^T(s\!-\!1,r\!-\!1)\textbf{G}(s\!-\!1,r\!-\!1)+ \textbf{G}^T(s\!-\!1,r)\textbf{G}(s\!-\!1,r)& \textbf{G}^T(s\!-\!1,r)\textbf{G}(s\!-\!1,r)\\
    \textbf{G}^T(s\!-\!1,r)\textbf{G}(s\!-\!1,r) & \textbf{G}^T(s\!-\!1,r)\textbf{G}(s\!-\!1,r)
    \end{bmatrix},\label{combinedmatrix3}\\
    & = \begin{bmatrix}
    \widetilde{\textbf{G}}(s-1,r-1)+ \widetilde{\textbf{G}}(s-1,r)& \widetilde{\textbf{G}}(s-1,r)\\
    \widetilde{\textbf{G}}(s-1,r) & \widetilde{\textbf{G}}(s-1,r)
    \end{bmatrix}.\label{combinedmatrix4}
\end{align}
\hrule
\end{figure*}

Note that the encoder described by \eqref{robustcodeword} utilizes the construction presented in Section \ref{encoder}, which achieves threshold security parameter $t=2^{s-r}-1$, and we use the same choice of indices dedicated for the key and the message that results in a \textit{proper} code. 

\subsubsection{Decoder}
We present a unified SC superdecoder for the coding scheme described above that corrects $\rho \leq D_{\mathrm{min}}-1$ erasures, where $D_{\mathrm{min}}=2^{s-r}$, and recovers the message given the shared key. The recursive decoder takes the received bit sequence $y_1^n$, the shared key $\bk$, key indices $\mathcal{A}^c$, code parameters $s,r$, and a recursion parameter $j$ as inputs. Initially, $j=1$. It outputs $h_1^n$, i.e., which is equal to the codeword $\bc$ provided that $\rho \leq D_{\mathrm{min}}-1$, as well as $u_1^n=\pi(\bk,\bm)$, which is used to retrieve the message $\bm$, and a recursion index $j'$ used to track the index of the last decoded bit. A pseudocode for the decoder is shown in Algorithm \ref{decoderrecursiveBEC}. The following claim shows the success of the described decoder.

\begin{algorithm}[t]
\caption{Unified SC decoder for binary erasures (DecBE)}
\begin{algorithmic}[1]\label{decoderrecursiveBEC}
\STATE Input: $\bk$, $y_1^n$, $\mathcal{A}^c$, $s$, $r$, $j$.
\STATE Output: $h_1^n$, $u_1^n$, $j'$.
    \IF{$r=0$}
        \STATE $\mathcal{I} = [j,j+1,....,j+2^{s}-1]$
        \STATE $i_1 \leftarrow$ index of any non-erasure bit in $y_1^n$.
        \FOR{$i \in \mathcal{A}^c$}
            \STATE $u_{i}=k_{i}$
        \ENDFOR
        \STATE $i' \in \mathcal{I} \setminus \mathcal{A}^c$
        \STATE $u_{i'}=y_{i_1} \oplus_{i \in \mathcal{A}^c} u_i$
        \STATE $h_1^n=[y_{i_1}, y_{i_1},\mydots, y_{i_1}]$
        \STATE $j'=j+2^s-1$
    \ELSE
        \STATE $\Bar{\by}=y_{1}^{n/2}\oplus y_{n/2+1}^n$
        \STATE $\bh_1', u_1^{n/2}, j_1' \leftarrow$ DecBE($\bk_1,\Bar{\by},\mathcal{A}_1^c,s\!-\!1,r\!-\!1,j$)
        \STATE $\bh_2'=u_1^{n/2} \widetilde{\textbf{G}}(s-1,r)$
        \STATE $\bh'=[\bh_1' \oplus \bh_2', \bh_2']$
        \STATE $\Tilde{y}_1^{n}=y_{1}^{n}\oplus \bh'=[\Tilde{\by_1}, \Tilde{\by_2}]$
        \STATE $l=\argminl_{j\in{1,2}}$ (number of erasures in $\Tilde{\by_j}$)
        \STATE $\bh_1'', u_{n/2+1}^{n}, j' \leftarrow$ DecBE($\bk_2, \Tilde{\by}_l, \mathcal{A}_2^c, s\!-\!1,r,j_1'+1$) 
        \STATE $\bh''=[\bh_1'',\bh_1'']$
        \STATE $h_1^n=\bh' \oplus \bh''$
    \ENDIF
    \RETURN $u_1^n$, $h_1^n$, $j'$
\end{algorithmic}
\end{algorithm}

\begin{claim}
The proposed unified RM-based coding scheme together with the unified SC superdecoder in Algorithm \ref{decoderrecursiveBEC} successfully retrieves the message as long as $\rho \leq D_{\mathrm{min}}-1$.
\end{claim}

\begin{proof}
Let the received sequence be denoted by $y_1^n$ which has at most $\rho$ erasures. Let also the key bits be denoted by $\bk$ which are assigned to entries of $\bu$ indexed by elements of $\mathcal{A}^c$. We use induction on the parameter $s$ of the underlying RM code of length $2^s$ to prove the claim. The induction hypothesis is that the decoder is successful for any RM-based coding scheme of length $2^s$ with some parameter $r \leq s$, and a key with size $\sum_{i=r+1}^s \binom{s}{i}$, assuming there are at most $\rho=2^{s-r}-1$ erasures. The induction base is $s=0$, for which the induction hypothesis is trivial. Now, suppose that the induction hypothesis holds for $s$ and we want to show it for $s+1$.

\textbf{Case 1:} $r=0$, i.e., we have an RM($s+1,0$) which becomes a repetition code of length $n = 2^{s+1}$. In this case, $\widetilde{\textbf{G}}(s+1,0)$ is the all-ones matrix and the entries of codeword are all equal to the sum of entries in $\bu$. Note that the number of message bits is $m = \sum_{i=0}^{r}\binom{s+1}{i} = 1$ and we have $2^{s+1}-1$ key bits. Also, the maximum number of erasures the code can correct is $2^{s+1}-1$. Hence, the decoder successfully retrieves the message bit using the non-erasure symbols, which there is at least one, in $y_1^n$. Suppose that the non-erasure bit is indexed by $i_1$. Since the locations of the key bits are known, we can place them at their respective locations retrieving $u_i$'s for all $i \in \mathcal{A}^c$. Next, the message bit located at $i'$ is retrieved as $u_{i'}=y_{i_1} \oplus_{i \in \mathcal{A}^c} u_i$, and the corresponding codeword is also retrieved correctly. Hence, the decoder is successful. Note that this case corresponds to lines 4-12 of Algorithm \ref{decoderrecursiveBEC}. 

\textbf{Case 2:} $r>0$. The code length is $n=2^{s+1}$ and the key length is $\sum_{i=r+1}^{s+1} \binom{s+1}{i}$. We split the key indices into two parts, namely $\mathcal{A}_1^c$ and $\mathcal{A}_2^c$, representing the key bits $\bk_1$ and $\bk_2$ in the first and the second half sub-blocks of $\bu$, respectively. The lengths of $\bk_1$ and $\bk_2$ are $|\mathcal{A}_1^c|=\sum_{i=r}^s \binom{s}{i}$ and $|\mathcal{A}_2^c|=\sum_{i=r+1}^s \binom{s}{i}$, respectively, due to the aforementioned choice of indices. The decoder first computes $\Bar{\by}=y_{1}^{n/2}\oplus y_{n/2+1}^n$ which will have at most $2^{s+1-r}-1$ erasures. It then passes this to the left child, in the binary tree representation terminology discussed earlier, along with $\bk_1$ and the set of its corresponding indices $\mathcal{A}_1^c$. The left child decodes a codeword of length $n'=2^s$ using a code with parameter $r'=r-1 \geq 0$, which can correct up to $2^{s-r'}-1=2^{s+1-r}-1$ erasures, and retrieves the message bits in $u_1^{n/2}$ given the key $\bk_1$ of length $\sum_{i=r'+1}^s \binom{s}{i}=\sum_{i=r}^s \binom{s}{i}$. The decoder on the left child is successful by induction hypothesis. It outputs $u_1^{n/2}$ and $\bh_1'$. After that, the decoder computes $\bh_2'= u_1^{n/2}\widetilde{\textbf{G}}(s,r)$ followed by $\bh'=[\bh_1' \oplus \bh_2', \bh_2']$. Then, the decoder computes $\Tilde{y}_1^{n}=y_1^{n} \oplus \bh'$ and chooses either $\Tilde{y}_1^{n/2}$ or $\Tilde{y}_{n/2+1}^{n}$, whichever has a smaller number of erasures, and passes it to the right child together with $\bk_2$ and the corresponding set of indices $\mathcal{A}_2^c$. The number of erasures in what is passed to this child is at most $2^{s+1-r}/2-1=2^{s-r}-1$, and the length of the key is $\sum_{i=r+1}^s \binom{s}{i}$. The decoder on this child decodes a codeword of length $n'=2^s$ using a code with parameter $r'=r > 0$, which can correct up to $2^{s-r'}-1=2^{s-r}-1$ erasures and retrieves the message bits in $u_{n/2+1}^{n}$ using the key $\bk_2$ of length $\sum_{i=r'+1}^s \binom{s}{i}=\sum_{i=r+1}^s \binom{s}{i}$. Decoding here is also successful by induction hypothesis. It outputs $u_{n/2+1}^n$ and $\bh_1''$. The overall decoder then computes $\bh''=[\bh_1'',\bh_1'']$ and outputs $h_1^n=\bh' \oplus \bh''$ and $u_1^n$. Hence, $u_1^n$ is retrieved and the proof is complete.
\end{proof}

\section{Conclusion}\label{conclusion}
In this work, we propose a model for threshold-secure coding with a shared key such that specific conditions for reliability and security based on information-theoretic measures are met. The specification of such model includes a threshold parameter which is to be designed based on the application for such coding schemes. Also, methods for utilizing error-correcting linear block codes in constructing threshold-secure coding schemes are discussed, where the parameter $t$ of the threshold-secure scheme is shown to be directly related to the minimum distance of the underlying linear block code. Furthermore, a coding scheme based on Reed-Muller codes is described. Its encoding is done recursively and is shown to satisfy the conditions for a \textit{proper} code. Moreover, a setup taking into account the noise in the communication channel between legitimate parties is considered. Then, a robust and threshold-secure coding scheme, based on code concatenation, is suggested for general channels. Also, a unified coding scheme built upon Reed-Muller codes for both threshold security and robustness in the presence of erasures is described.

A possible direction for future work is to design coding schemes based on punctured Reed-Muller codes to allow for more flexible rates. To this end, ideas from punctured schemes for closely related polar codes can be useful \cite{shin2013design,el2015harq}. Also, it is interesting to explore whether unified coding schemes for threshold security and robustness, similar to the RM-based scheme presented in Section\,\ref{robustlowcomplexity}, can be constructed from other well-known families of codes. Another possible direction of future work is to study threshold security in settings with wiretap channels, where the eavesdropper's channel is also noisy. Also, extending the considered setup to multi-user scenarios, as in wiretap multiple access \cite{tekin2008gaussian} or as in multi-user secret sharing setups \cite{soleymani2020distributed}, is another interesting future direction.

\bibliographystyle{IEEEtran}
\bibliography{IEEEabrv}

\end{document}